\documentclass[10pt]{article}
\usepackage[a4paper,left=3cm,right=3cm]{geometry}
\usepackage[utf8]{inputenc}
\usepackage[T1]{fontenc}
\usepackage[symbol]{footmisc}
\usepackage{authblk}
\usepackage{siunitx}
\usepackage{booktabs}
\usepackage{amsmath,amssymb,amsthm}
\usepackage{cases}
\usepackage{url}
\usepackage{bm}
\usepackage{graphicx}
\usepackage{enumerate}
\usepackage{algorithm}
\usepackage{caption}
\usepackage{mathtools}
\usepackage{todonotes}
\usepackage{scalerel}

\newtheorem{thm}{Theorem}[section]
\newtheorem{cor}[thm]{Corollary}
\newtheorem{lem}[thm]{Lemma}

\theoremstyle{definition}
\newtheorem{defn}[thm]{Definition}
\theoremstyle{remark}
\newtheorem{rem}[thm]{Remark}
\numberwithin{equation}{section}

\newcommand{\E}[1]{\mathbb{E}\left[#1\right]}

\newcommand{\RR}{\mathbb{R}}

\newcommand{\PP}{\mathbb{P}}

\newcommand{\CC}{\mathbb{C}}

\newcommand{\cO}{\mathcal{O}}

\newcommand{\set}[1]{\left\{#1\right\}}

\newcommand{\ind}[1]{\bm{1}_{#1}}

\renewcommand{\Re}{\mathrm{Re}}

\newcommand{\pp}{\textup{\texttt{+}}}
\newcommand{\mm}{\textup{\texttt{-}}}

\DeclareMathOperator{\arsinh}{arsinh}

\begin{document}
\pagestyle{empty}

\title{W-shaped implied volatility curves in a variance-gamma mixture model}
\author{Martin Keller-Ressel}
\affil{TU Dresden, Institute for Mathematical Stochastics, Dresden, 01062, Germany}
\affil{martin.keller-ressel@tu-dresden.de}

\maketitle

\begin{abstract}
In liquid option markets, W-shaped implied volatility curves have occasionally be observed. We show that such shapes can be reproduced in a mixture of two variance-gamma models. This is in contrast to lognormal models, where at least three different distributions have to be mixed in order to produce a W-shape, as recently shown by Glasserman and Pirjol.
\end{abstract}


\section{Introduction}
W-shaped implied volatility curves have occasionally been observed in liquid option markets \cite{voladyn2022market} and have been linked to diverging investor expectations in anticipation of high-impact events, such as earning announcements. A comprehensive mathematical treatment of W-shaped volatility curves has recently been given by Glasserman and Pirjol in \cite{glasserman2021w}. Among other findings, they show that there is no simple link between bimodality of risk-neutral densities and W-shaped implied volatilities, and they study the appearance of W-shaped volatility curves in lognormal mixture models. One striking result \cite[Cor.~4.1]{glasserman2021w} is that in such a mixture model, mixing \emph{two} lognormal densities is not sufficient to produce a W-shaped volatility curve, and that at least \emph{three} such densities have to be mixed in order to produce a W-shape. This seems at odds with the explanation of `diverging investor beliefs', where just two (bearish and bullish) groups of investors should suffice to produce a W-shaped implied volatility curve. Here, we show that when the lognormal family of distributions is replaced by the variance-gamma family, a two-distribution mixture is sufficient to produce W-shapes in implied volatility curves. This links the study of W-shaped volatility curves to the literature on non-Gaussian and L\'evy-process-based option pricing models, as pioneered by \cite{madan1990variance, madan1998variance} and shows that such models can lead to more parsimonious explanations for unusual shapes than Gaussian models.

\section{Sufficent conditions for W-shaped smiles}

We fix a time horizon $T > 0$ and consider a price process $S$ with log-price $X_t = \log(S_t/S_0)$. We assume that $X_T$ has a continuous density $f$, satisfying the risk-neutrality condition $\int_{-\infty}^\infty e^x f(x) dx = 1$. The corresponding density $g$ of $S_T$ is given by $g(s) = \tfrac{1}{s}f(\log(s/S_0))$,
and to $g$ we associate the call-price function
\begin{equation}\label{eq:call}
C(K) = \int_K^\infty (s - K) g(s) ds.
\end{equation}
The density $g$ can be recovered from $C$ by differentiating twice, i.e., we have
\begin{equation}\label{eq:density}
C''(s) = g(s) \qquad \text{for all $s \in (0,\infty)$.}
\end{equation}
Note that the continuity of $g$ has allowed us to strengthen this equality from holding `almost everywhere' to `everywhere'.
Denoting by $C_{BS}(K,\sigma)$ the Black-Scholes price of a call with time-to-maturity $T$ and initial stock price $S_0$, the implied volatility $\sigma(K)$ associated to $g$ (or equivalently to $f$) is given as the unique solution of 
\begin{equation}\label{eq:IV}
C(K) = C_{BS}(K,\sigma(K)).
\end{equation}
By the inverse function theorem, $\sigma$ is a smooth function of $C$, and hence at least twice continuously differentiable. We denote by 
\[\varphi_\sigma(x) = \frac{1}{\sqrt{2\pi}\sigma} \exp\left(-\tfrac{1}{2}\left(\tfrac{x}{\sigma} + \tfrac{\sigma}{2}\right)^2\right) \]
the density of a normal distribution with variance $\sigma^2$ and mean $-\sigma^2/2$. This density is the risk-neutral log-price density of the Black-Scholes model with volatility $\sigma$.

\begin{defn}
\begin{enumerate}[(a)]
\item The implied volatility smile $K \mapsto \sigma(K)$ is called \textbf{W-shaped}, if there exists $\sigma_* \in (0,\infty)$, such that
\[K \mapsto \sigma(K) - \sigma*\]
has exactly four sign changes with sign sequence $\pp\mm\pp\mm\pp$; cf. \cite[Def.~2.1]{glasserman2021w}.
\item The implied volatility smile $K \mapsto \sigma(K)$ is called \textbf{(W+)-shaped}, if there exists $\sigma_* \in (0,\infty)$, such that 
\[K \mapsto \sigma(K) - \sigma*\]
has as an even number $d$ of sign changes with $d \ge 4$ and with sign sequence starting and terminating with $\pp$.
\end{enumerate}
\end{defn}
We remark that a `sign change' always refers to a 'strong sign change', i.e. a change from a strictly positive sign to a strictly negative sign and vice versa.\\ 

We now give sufficient conditions for a (W+)-shaped smile, which -- up to condition (b) -- have been implicitly considered, without giving them names, already in \cite[Sec.~6.4]{glasserman2021w}.

\begin{lem}\label{prop:sufficient}
Suppose that the continuous log-price density $f$ satisfies the following three conditions:
\begin{enumerate}[(a)]
\item \textbf{geometric symmetry:}
\begin{equation}\label{eq:geo_sym}
e^{x/2} f(x) = e^{-x/2} f(-x), \qquad \forall\,x \in \RR.
\end{equation}
\item \textbf{semi-heavy tails:} 
\begin{equation}\label{eq:semi_heavy}
r_* := \sup \set{r: \int_{-\infty}^\infty e^{rx} f(x)dx < \infty}  < \infty.
\end{equation}
\item \textbf{dip-at-zero:}
\begin{equation}\label{eq:dip}
f(0) < \varphi_{\sigma(S_0)}(0) = \frac{1}{\sqrt{2 \pi} \sigma(S_0)}e^{-\sigma(S_0)^2 / 8}.
\end{equation}
\end{enumerate}
Then the associated implied volatility smile $K \mapsto \sigma(K)$ is (W+)-shaped.
\end{lem}
A proof can essentially be distilled from \cite[Sec.~6.4]{glasserman2021w}. We first discuss the conditions in more detail, and then give a self-contained proof of the result.
\begin{itemize}
\item By \cite{carr2009put} the first condition is equivalent to the symmetry of the implied volatility smile in the sense that
\begin{equation}\label{eq:sym}
\sigma(K) = \sigma(S_0^2 / K), \qquad \forall\,K \in (0,\infty).
\end{equation}
It also automatically implies the risk-neutrality condition $\int e^x f(x) dx = 1$.
\item By Lee's moment formula (see \cite{lee2004moment, benaim2008smile}), the second condition implies that $\lim_{K \to \infty} \sigma(K) = \infty$. Together with condition (a) it follows that also $\lim_{K \to 0} \sigma(K) = \infty$
\item The last condition, together with (a), implies concavity of the ATM implied volatility, i.e., $\sigma''(S_0) <0$, as shown below and in \cite[Sec.~6.4]{glasserman2021w}, who give a similar condition in the context of log-normal mixture models. Note that the density $f$ appears on the left side of \eqref{eq:dip}, but also indirectly on the right side, since the implied volatility $\sigma(S_0)$ also depends on $f$. 
\end{itemize}

\begin{proof}
By the symmetry property \eqref{eq:sym} it follows that $\sigma'(S_0) = 0$. Taking second derivatives of \eqref{eq:IV} and evaluating at $K = S_0$, we obtain
\[C''(S_0) = \partial_{KK}C_{BS}(S_0, \sigma(S_0)) + \partial_{\sigma}C_{BS}(S_0, \sigma(S_0)) \sigma''(S_0).\]
Rearranging yields
\[\sigma''(S_0) = \frac{C''(S_0) -  \partial_{KK}C_{BS}(S_0, \sigma(S_0))}{\partial_{\sigma}C_{BS}(S_0, \sigma(S_0))},\]
see also \cite[Eq.~(79)]{glasserman2021w}. In the denominator, we find the option's Vega, wich is always strictly positive. The first term in the numerator is equal to $f(0)$, by
\eqref{eq:density}. The second term is the option's Gamma, given by 
\[\partial_{KK}C_{BS}(S_0, \sigma(S_0)) = \phi_{\sigma(S_0)}(0) = \frac{1}{\sqrt{2 \pi} \sigma(S_0)}e^{-\sigma(S_0)^2 / 8}. \]
Thus, the dip-at-zero condition \eqref{eq:dip} implies that $\sigma''(S_0) < 0$. Together with $\sigma'(S_0) = 0$ we conclude that $K \mapsto \sigma(K)$ has a local maximum at $K = S_0$. By continuity, there exists $\sigma_* < \sigma(S_0)$, such that $K \mapsto \sigma(K) - \sigma_*$ has (at least) two sign changes, one to the left of $S_0$ and one to the right of $S_0$. Going from $S_0$ to the right the first sign change must be from $\pp$ to $\mm$. But $\lim_{K \to \infty} \sigma(K) = \infty$, such that there must be (at least) one other sign change from $\mm$ to $\pp$; the terminal sign being $\pp$. We conclude that there are at least two sign changes to the right of $S_0$, ending in $\pp$. By symmetry, the same holds to the left of $S_0$, and we conclude that the smile must be (W+)-shaped.
\end{proof}

\section{A variance-gamma mixture model}
We introduce a mixture of two variance-gamma processes, which -- for suitable choice of parameters -- satisfies the conditions of Proposition~\ref{prop:sufficient} and hence produces a (W+)-shaped implied volatility curve at time $T >0$. With some additional effort we will then show that the curve is actually W-shaped. Variance-gamma processes belong to the best-known non-Gaussian models for option pricing, see e.g. \cite{madan1990variance, madan1998variance}.
\subsection{The two components}
We fix positive parameters $v, c, \lambda, \mu$ such that $\mu < 2\lambda$. Our starting point are the two drifted Brownian motions\footnote{For the following construction it does not matter whether the same or two independent Brownian motions are used for $Y^+$ and $Y^-$.}
\begin{align}
Y^\pm_t &= v B_t - \tfrac{v^2}{2}t \pm \mu t.
\end{align}
The two different drift rates $\pm \mu$ can be interpreted as the diverging beliefs of two groups of investors on the future stock price evolution.
As stochastic time-changes we use two Gamma processes $L^\pm$ with shape parameter $c$ and rate parameters $\lambda \pm \tfrac{\mu}{2} > 0$, cf. \cite[Ch.~4]{tankov2003financial}. The differing rate parameters will allow us to match the geometric symmetry condition in the mixture model. In terms of marginal distributions, we have 
\[L^\pm_t \sim \Gamma(ct,\lambda \pm \tfrac{\mu}{2}),\]
where $\Gamma(a,b)$ denotes the Gamma distribution in shape-rate parameterization. The characteristic exponents $\ell_\pm(u) = \tfrac{1}{t} \log \E{e^{uL^\pm_t}}$  of the processes $L^\pm$ (equivalently, the cumulant generating function of $L_1^\pm$) are given by
\[\ell_\pm(u) = - c \log\left(1 - \tfrac{u}{\lambda \pm \mu/2}\right), \qquad u \in (-\infty, \lambda \pm \mu/2).\]
Note that the processes $L^\pm$ can also be derived from a single $\Gamma(c,\lambda)$-process by applying an exponential tilt of $\mp \tfrac{\mu}{2}$ respectively, i.e., we can write 
\begin{equation}\label{eq:tilt}
\ell_\pm(u) = \ell(u \mp \tfrac{\mu}{2}) - \ell(\mp \tfrac{\mu}{2}) \qquad \text{where} \qquad \ell(u) = -c \log\left(1 - \tfrac{u}{\lambda}\right).
\end{equation}
Using $L^\pm$ as time-changes applied to $Y^\pm$ we arrive at the variance-gamma processes
\begin{equation}\label{eq:VG}
X_t^\pm = Y^\pm_{L^\pm_t} = v B_{L_t^\pm} + \left(\pm \mu - \tfrac{v^2}{2}\right) L^\pm_t,
\end{equation}
with characteristic exponents
\begin{equation}\label{eq:psi}
\psi_\pm(u) = \ell_{\pm}\left(\tfrac{v^2}{2}(u^2 - u) \pm \mu u\right).
\end{equation}
These characteristic exponents can be transformed into the standard $(\sigma, \theta, \kappa)$-parameterization of the variance-gamma process, cf. \cite[Ch.~4]{tankov2003financial}
\[\psi_{VG}(u) = -\frac{1}{\kappa} \log\left(1 - \tfrac{u^2 \sigma^2 \kappa}{2} - \theta \kappa u \right),\]
by setting 
\[\sigma^2 = c\frac{v^2}{\lambda \pm \mu/2}, \qquad \theta = c\frac{-v^2/2 \pm \mu}{\lambda \pm \mu/2}, \qquad \kappa = 1/c.\]
Furthermore, the densities $f_\pm$ of $X_t^\pm$ are given by 
\begin{equation}\label{eq:VG_density}
f_\pm(x) = \gamma_\pm \left|\tfrac{x}{\alpha}\right|^{ct - 1/2} e^{\beta_\pm x} K_{ct - 1/2}(\alpha |x|), \qquad x \in \RR
\end{equation}
where $K_\nu(x)$ denotes the modified Bessel function of the second kind, and
\begin{equation}\label{eq:VG_par}
\alpha = \sqrt{\tfrac{\mu^2}{v^4} + \tfrac{2\lambda}{v^2} + \tfrac{1}{4}}, \qquad \beta_\pm = \pm \tfrac{\mu}{v^2} - \tfrac{1}{2} , \qquad \gamma_\pm = \sqrt{\tfrac{2}{\pi}} \tfrac{1}{\Gamma(ct)} \left(\tfrac{\lambda \pm \mu/2}{v^2}\right)^{ct}.
\end{equation}
This gives the alternative $(\alpha, \beta, c)$-parameterization of the variance-gamma distribution, which can be considered a special case of the $(\alpha, \beta, \delta, \lambda, \mu)$-parameterization of the generalized hyperbolic distribution, where $\lambda = c$ and $\delta = \mu = 0$, see \cite[Eq.~(4.38)]{tankov2003financial} and \cite{scott2018variance}.

\subsection{Mixing the components}
Our next goal is to mix $X_T^\pm$ such that the geometric symmetry condition \eqref{eq:geo_sym} is satisfied. To this end, we fix $T > 0$ and set

\begin{equation}\label{eq:ab}
a = \exp(-T \ell(\tfrac{\mu}{2})) = \left(1 + \tfrac{\mu}{2\lambda}\right)^{cT}, \qquad b = \exp(-T\ell(-\tfrac{\mu}{2})) = \left(1 - \tfrac{\mu}{2\lambda}\right)^{cT}.
\end{equation}
Let now $M$ be a Bernoulli random variable, independent of $X^\pm$, taking the value $1$ with probability 
\[p = \PP(M = 1) = \tfrac{a}{a+b} = \left(\tfrac{\lambda}{\mu} + \tfrac{1}{2}\right)^{cT} \Big/ \left\{\left(\tfrac{\lambda}{\mu} + \tfrac{1}{2}\right)^{cT} + \left(\tfrac{\lambda}{\mu} - \tfrac{1}{2}\right)^{cT}\right\}\]
and the value $0$ with complementary probability $1-p = b/(a+b)$. 
The log-price $X_T = \log(S_T/S_0)$ of our variance-gamma mixture model, is defined as the mixture
\begin{equation}\label{eq:mixture}
X_T = M\cdot X^-_T + (1 - M) X^+_T.
\end{equation}
\begin{lem}\label{lem:mixture}The log-price $X_T$ in the variance-gamma mixture model \eqref{eq:mixture} has the continuous density
\begin{equation}\label{eq:mix_dens}
f(x) = \tfrac{a}{a+b} f_-(x) + \tfrac{b}{a+b} f_+(x),
\end{equation}
with $f_\pm$ given by \eqref{eq:VG_density}. The density $f$ satisfies the geometric symmetry condition \eqref{eq:geo_sym} and the semi-heavy tail condition \eqref{eq:semi_heavy}. Moreover, the moment generating function of $X_T$ is 
\[m(u) = \tfrac{ab}{a+b}\left\{\exp\left(T \ell\left(\tfrac{v^2}{2}(u^2 - u) - \mu u + \tfrac{\mu}{2}\right)\right) + \exp\left(T \ell\left(\tfrac{v^2}{2}(u^2 - u) + \mu u - \tfrac{\mu}{2}\right)\right)\right\}.\]

\end{lem}
\begin{proof}
The form of the density follows from the definition \eqref{eq:mixture} of the mixture model. The form of the moment generating function follows by combining \eqref{eq:tilt} with \eqref{eq:psi}. To verify the geometric symmetry condition, combining \eqref{eq:VG_density} with \eqref{eq:ab} yields the relations
\begin{equation}\label{eq:ab_sym}
\begin{split}
b e^{x/2} f_+(x) &= a e^{-x/2} f_-(-x)\\
a e^{x/2} f_-(x) &= b e^{-x/2} f_+(-x).
\end{split}
\end{equation}
Adding both equations and dividing by $a+b$ gives \eqref{eq:geo_sym}. Alternatively, note that the moment generating function satisfies $m(u) = m(1-u)$ which is equivalent to the geometric symmetry condition; cf. \cite{carr2009put}. For the semi-heavy tail condition, note that $\lim_{u \to \pm \infty} m(u) = \infty$, by the properties of $\ell$, see \eqref{eq:tilt}.
\end{proof}
To apply Proposition~\ref{prop:sufficient} to the variance-gamma mixture model it remains to show that also the dip-at-zero condition \eqref{eq:dip} can be satisfied. For this, analyzing the limiting case $v \to 0$ will be crucial. 

\begin{center}
\includegraphics[width=0.65\textwidth]{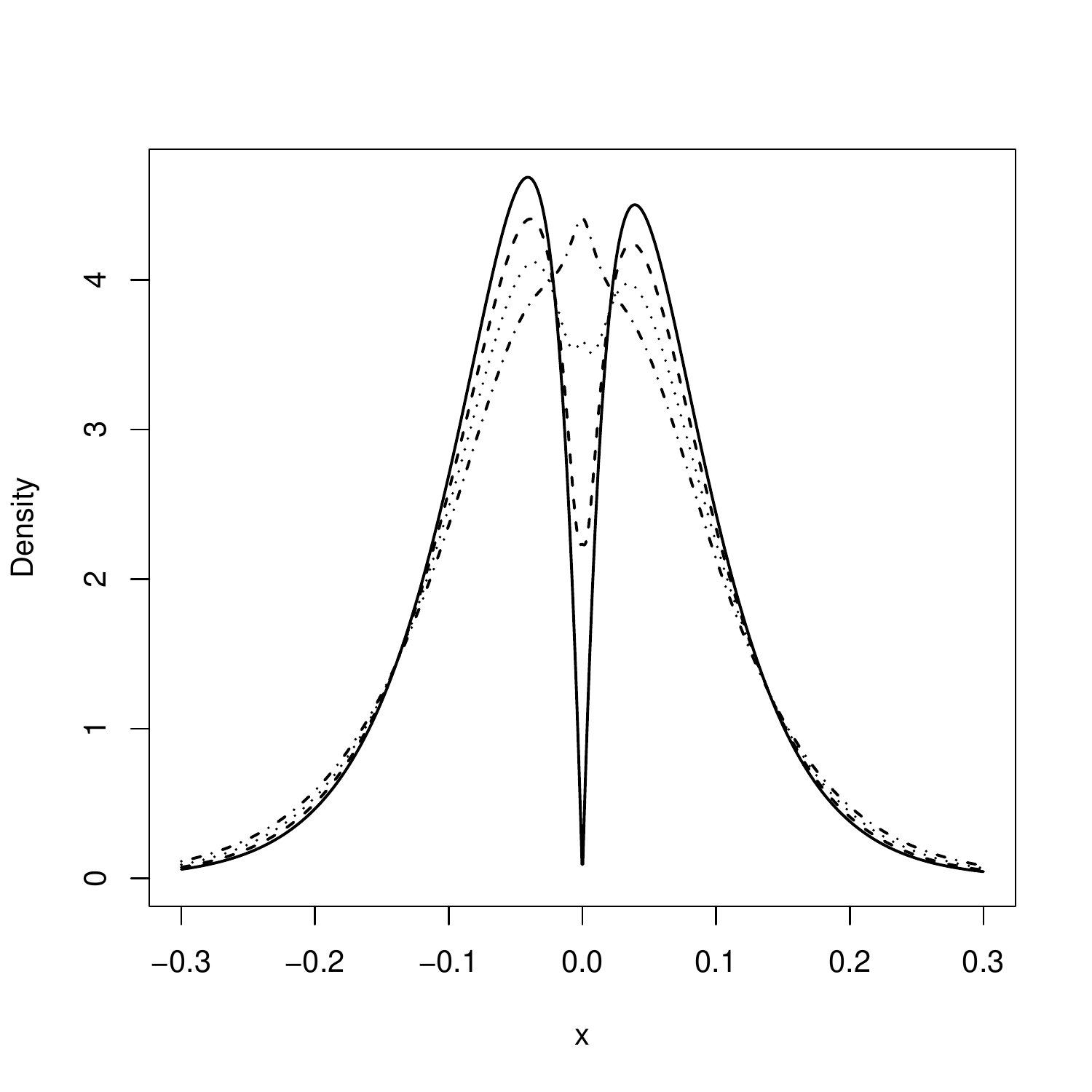}
\captionof{figure}{\textbf{Densities of the variance-gamma mixture model.} This plot shows the densities $f_v$ of the variance-gamma mixture model (see Lemma~\ref{lem:mixture}) with parameters $cT = 2$, $\lambda = 0.5$, $\mu = 0.02$ and $v = (0, 0.01, 0.015, 0.02)$. The solid line corresponds to the case $v = 0$, where $f_0$ becomes a double gamma density.}
\label{fig:density}
\end{center}

\subsection{The limit $v \to 0$}
We indicate the dependency of various quantities on $v$ by writing $f_v(x), m_v(u), X_{T,v}$, etc. To see what happens as $v \to 0$, consider the time-change equation \eqref{eq:VG}. As $v \to 0$ the Brownian motion (and the drift $-\tfrac{v^2}{2}$) vanish completely, and we are left with
\begin{equation}\label{eq:X0}
X_{t,0}^\pm := \lim_{v \to 0} X_{t,v}^\pm  = \pm \mu L_t^\pm\qquad a.s.
\end{equation}
By the scaling law of the Gamma distribution, this means that
\[X_{t,0}^+ \sim \Gamma\left(ct,\tfrac{\lambda}{\mu} + \tfrac{1}{2}\right)\quad \text{and} \quad -X_{t,0}^- \sim \Gamma\left(ct,\tfrac{\lambda}{\mu} - \tfrac{1}{2}\right).\]
Hence, $X_{t,0}$ becomes a mixture of an ordinary Gamma distribution (supported on $(0,\infty)$) and a reflected Gamma distribution (supported on $(-\infty,0)$). Distributions of this type have occasionally been studied under the name of double gamma distributions, cf. \cite{rao2012estimation}. Note, however, that even the almost sure convergence of random variables in \eqref{eq:X0} does not guarantee the pointwise convergence of densities, which is of interest here, in light of condition \eqref{eq:dip}. 

\begin{lem}\flushleft \begin{enumerate}[(a)]
\item The pointwise convergence $\lim_{v \to 0} f_v(x) = f_0(x)$ holds for all $x \in \RR$, where $f_0$ is an asymmetric double gamma density, given by
\begin{equation}\label{eq:double_gamma}
f_0(x) = \frac{a}{a+b} \frac{\lambda_-^{cT}}{\Gamma(cT)} (-x)^{cT-1} e^{\lambda_- x} \ind{(-\infty,0)}(x)  + \frac{b}{a+b} \frac{\lambda_+^{cT}}{\Gamma(cT)} x^{cT-1} e^{-\lambda_+ x} \ind{(0,\infty)}(x)
\end{equation}
with $\lambda_\pm := \lambda/\mu \pm 1/2$. 
\item Let $cT > 1$. Then $f_0$ is continuous, in particular at $x = 0$, where $f_0(0) = 0$. Moreover, $f_0$ is bimodal with modes at $-(cT - 1)/\lambda_- < 0$ and $(cT-1)/\lambda_+ > 0$.  
\end{enumerate}
\end{lem}
\begin{rem}
\begin{enumerate}[(a)]
\item We remark that the (symmetric) double gamma distribution appears in \cite{balanda1988kurtosis} as an example of a bimodal and `peak-less' distribution with large kurtosis, refuting an earlier claim made by Darlington \cite{darlington1970kurtosis} that bimodal distributions must have small kurtosis. The density is shown in Figure~\ref{fig:density}.
\item A stronger convergence result for the (analytically extended) densities $f_v$ is shown in Lemma~\ref{lem:analytic}.
\end{enumerate}
\end{rem}
\begin{proof}
We start by showing the convergence of $f_{+,v}(x)$ to the second term of \eqref{eq:double_gamma}. By \cite[9.7.2]{abramowitz1964handbook} the modified Bessel function of the second kind has the asymptotic behavior 
\[e^{-z}K_\nu(z) = \sqrt{\tfrac{\pi}{2z}}\left(1 + \mathcal{O}(1/z)\right)\]
as $z \to \infty$. Together with 
\begin{equation}\label{eq:alpha_beta}
\begin{split}
\lim_{v \to \infty} (\alpha_{v} \mp \beta_{\pm,v}) &= \lim_{v \to \infty} \sqrt{\tfrac{\mu^2}{v^4} + \tfrac{2\lambda}{v^2} + \tfrac{1}{4}} \mp \left(\pm \tfrac{\mu}{v^2} - \tfrac{1}{2}\right) = \frac{\lambda}{\mu} \pm \frac{1}{2} = \lambda_\pm,\\
\lim_{v \to \infty} (\alpha_{v} \pm \beta_{\pm,v}) &= +\infty
\end{split}
\end{equation}
we obtain
\[f_{0,+}(x) := \lim_{v \to 0} f_{+,v}(x) = \frac{1}{\Gamma(cT)} \lambda_+^{cT} x^{cT - 1} \exp\left(- \lambda_+ x \right), \qquad \text{for $x > 0$},\]
and $f_{0,+}(x) :=\lim_{v \to 0} f_+(x) = 0$ for $x < 0$. For $x = 0$, it can be shown by separate calculation that $\lim_{v \to 0} f_+(0) = 0$. The calculations for $f_{-,v}$ are similar and therefore omitted.\\
For part (b), note that 
\[\lim_{x \downarrow 0} f_{0,+}(x) = \begin{cases}+\infty &\quad cT < 1\\ \lambda_+ &\quad cT = 1\\0 &\quad cT > 1\end{cases},\]
and similarly for $f_{0,-}$ when approaching from the left. Together this shows that $f_0$ is continuous at $x = 0$ with $f_0(0) = 0$ if (and only if) $cT > 1$. The bimodality and the location of modes follow from well-known properties of the Gamma distribution.
\end{proof}

\subsection{Existence of W-shaped implied volatility curves}

To show that the variance-gamma mixture model \eqref{eq:mixture} is able to produce (W+)-shaped implied volatility curves, we apply Proposition~\ref{prop:sufficient} and obtain the following:

\begin{lem}\label{lem:Wplus}
Given positive parameters $c, \lambda, \mu$ with $\mu < 2\lambda$ and $cT > 1$,  there exists $v_* > 0$, such that $f_v$ satisfies conditions (a) - (c) of Proposition~\ref{prop:sufficient} for all $v \in [0,v_*)$. Consequently, the implied volatility smile $K \mapsto \sigma_v(K)$ associated to $f_v$ must be $(W+)$-shaped for all $v \in [0,v_*)$.
\end{lem}
\begin{proof}
The symmetry condition and the semi-heavy tails condition have been verified in Lemma~\ref{lem:mixture} above. For the dip-at-zero condition, note that by continuity it holds that 
\[\lim_{v \to 0}\sigma_v(S_0) = \sigma_0(S_0) \in (0,\infty),\]
where $\sigma_0$ is the implied volatility associated to the double-Gamma density $f_0$. Subtracting $f(0)$ from both sides of \eqref{eq:dip} and sending $v \to 0$ we obtain
\[\lim_{v \to 0} \left\{f_v(0) - \frac{\exp\left(-\sigma_v(S_0)^2/8\right)}{\sqrt{2\pi}\sigma_v(S_0)} \right\} = - \frac{\exp\left(-\sigma_0(S_0)^2/8\right)}{\sqrt{2\pi}\sigma_0(S_0)} < 0.\]
Therefore the dip-at-zero condition is satisfied for $v=0$, and, by continuity, also in a right neighborhood $[0,v_*)$ of $0$. The conclusion follows by applying Proposition~\ref{prop:sufficient}.
\end{proof}

In order to obtain a true W-shape it remains to show that for any $\sigma_* > 0$ the function $K \mapsto \sigma(K) - \sigma_*$ can have at most four sign changes. To this end, we apply the following result from \cite{glasserman2021w} (slightly reformulated):

\begin{lem}[Prop.~2.1 of \cite{glasserman2021w}] \label{lem:crossings}
Let $n_\text{vol}(\sigma_*)$ be the number of sign changes of $K \mapsto \sigma(K) - \sigma_*$, and let $n_\text{pdf}(\sigma_*)$ be the number of crossings of the densities $f$ and $\phi_{\sigma_*}$. If $f$ and $\phi_{\sigma_*}$ differ on an interval, then $n_\text{vol}(\sigma_*) \le 
n_\text{pdf}(\sigma_*) - 2$.
\end{lem}

Thus, if the densities $f$ and $\phi_{\sigma_*}$ cross at most six times for any $\sigma_* > 0$, then the implied volatility curve associated to $f$ can be at most W-shaped. We can verify this property for the double gamma density:

\begin{lem}\label{lem:Gamma_cross}
For any choice of parameters, the double gamma density $f_0$, see \eqref{eq:double_gamma}, and the density of the normal distribution $\phi_\sigma$ cross at most six times.
\end{lem}
\begin{proof}
We decompose the double gamma density as $f_0(x) = \tfrac{a}{a+b}f_{0,-}(-x) + \tfrac{b}{a+b}f_{0,+}(x)$, 
where $f_{0,+}(x)$ and $f_{0,-}(x)$ are $\Gamma(cT,\lambda_\pm)$-densities. By symmetry, it is clearly sufficient to show that $h_+(x) := \tfrac{b}{a+b} f_{0,+}(x)$ crosses $\phi_\sigma$ at most three times on $(0,\infty)$. Since both $h_+$ and $\phi_\sigma$ have full support on $(0,\infty)$ we can equivalently show that the difference $\Delta(x) = \log h_+(x) - \log \phi_\sigma(x)$ has at most three zeros. The function $\Delta(x)$ can be decomposed as
\[\Delta(x) = a_0 + a_1 \log(x) + a_2 x + a_3 x^2, \]
where 
\[a_0 = \log\left(e^{\sigma^2/8} \sqrt{2\pi} \sigma \lambda_+^{cT} / \Gamma(cT) \right), \quad a_1 = cT - 1, \quad a_2 = -\frac{\lambda}{\mu}, \quad a_3 = \frac{1}{2\sigma^2}.\] 
The functions $(1, \log(x), x, x^2)$ form a Descartes system on $(0,\infty)$, see e.g. \cite{gasull2012chebyshev}. By the variation-diminishing property of Descartes systems (see \cite{karlin1968total, borwein1995polynomials}), the number of sign changes of $\Delta$ is bounded by the number of sign changes of $(a_0,a_1,a_2,a_3)$. We conclude that $\Delta$ has at most three sign changes, and therefore the functions $h_+$ and $\phi_\sigma$ cross at most three times on $(0,\infty)$. Applying a symmetric argument to $h_-$, the claim follows.
\end{proof}

Note that the sequence $(a_0,a_1,a_2,a_3)$ can have three sign changes only if $cT > 1$. In the case $cT < 1$ the number of sign changes is bounded by two. This yields the following corollary:

\begin{cor}
If $cT < 1$, then the implied volatility curve $K \to \sigma_0(K)$ cannot be W- or (W+)-shaped.
\end{cor}

\begin{center}
\includegraphics[width=0.75\textwidth]{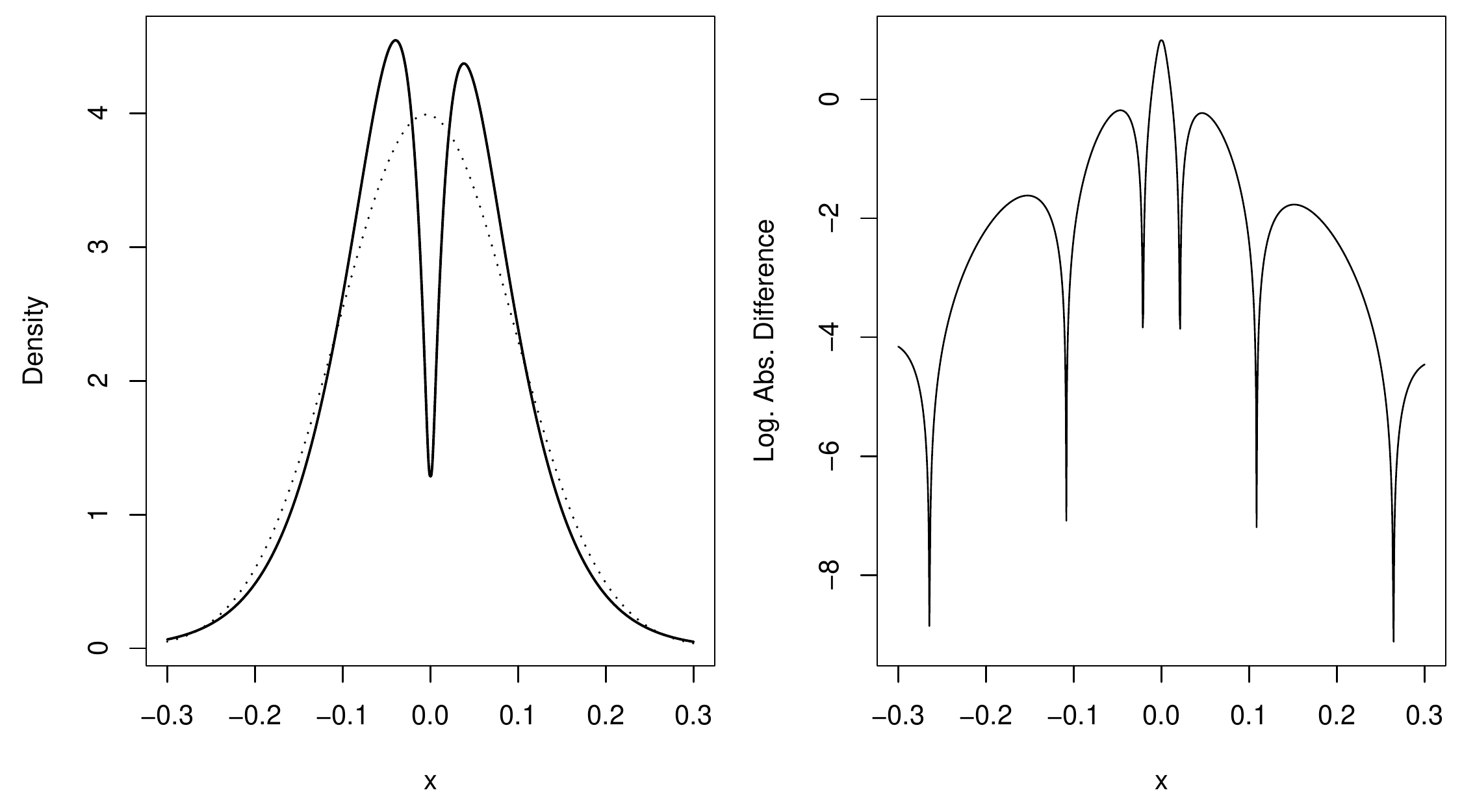}
\captionof{figure}{\textbf{double gamma vs. normal density.} The left panel shows a double gamma density (solid) with superimposed normal density (dotted). In the right panel the logarithm of the absolute difference of the densities is shown, indicating the six crossing points in line with Lemma~\ref{lem:Gamma_cross}.}
\end{center}

Combing the above results and using the auxilliary Lemma~\ref{lem:analytic}, we obtain our main result:

\begin{thm}
Given positive parameters $c, \lambda, \mu$ with $\mu < 2\lambda$ and $cT > 1$,  there exists $v_\circ > 0$, such that the implied volatility smile $K \mapsto \sigma_v(K)$ associated to $f_v$ is $W$-shaped for all $v \in [0,v_\circ)$.
\end{thm}
\begin{proof}By Lemma~\ref{lem:Wplus}, we know that there exists $v_*$ such that $f_v$ is (W+)-shaped for all $v \in [0, v_*)$. That is, it remains to show that there is $v_\circ \in (0,v_*]$, such that  $K \mapsto \sigma_v(K) - \sigma_*$ can not have \emph{more} than four zeroes for any $\sigma_* \in (0,\infty)$ and $v \in [0,v_\circ)$. By Lemma~\ref{lem:crossings} this follows, if we can show that the densities $f_v$ and $\phi_\sigma$ cross at most six times for any $\sigma$ and for $v \in [0,v_\circ)$. By Lemma~\ref{lem:Gamma_cross}, this statement is true in the boundary case $v=0$, and it remains to show that we can extend it to a right neighborhood $[0,v_\circ)$ of zero. To this end we apply a complex analytic argument. Consider, for arbitrary $\sigma > 0$, the difference
\[\delta_v(x) = f_v(x) - \phi_\sigma(x), \qquad x \in \RR.\]
By Lemma~\ref{lem:analytic}, $\delta_v$ can be analytically extended to the right complex half plane $H_+ = \set{u \in \CC: \Re\,u > 0}$ and to the left complex half plane $H_- = \set{u \in \CC: \Re\,u < 0}$, for any $v \ge 0$. Moreover, $\delta_0$ has three zeros on the positive real axis in $H_+$ and three on the negative real axis in $H_-$. Since zeroes of an analytic function are isolated, we can find open subsets $\cO_+ \subset H_+$ and $\cO_- \subset H_-$, containing the real half axes and the zeroes thereon, but no other (complex) zeroes of $\delta_0$. Also by Lemma~\ref{lem:analytic}, $\delta_v$ converges uniformly on compacts to $\delta_0$ as $v \to 0$ on both $\cO_+$ and $\cO_-$. We can now apply Hurwitz' theorem \cite[Sec.~IV.3]{markushevich2005theory} to conclude that there exists $v_\circ$ such that $\delta_v$ has the same number of zeros in $\cO_+$ (and $\cO_-$) as $\delta_0$ for all $v \in [0,v_\circ)$. In particular, this bounds the number of real zeroes of $\delta_v$ by six; together with Lemma~\ref{lem:crossings} this completes the proof.
\end{proof}

\begin{center}
\includegraphics[width=0.75\textwidth]{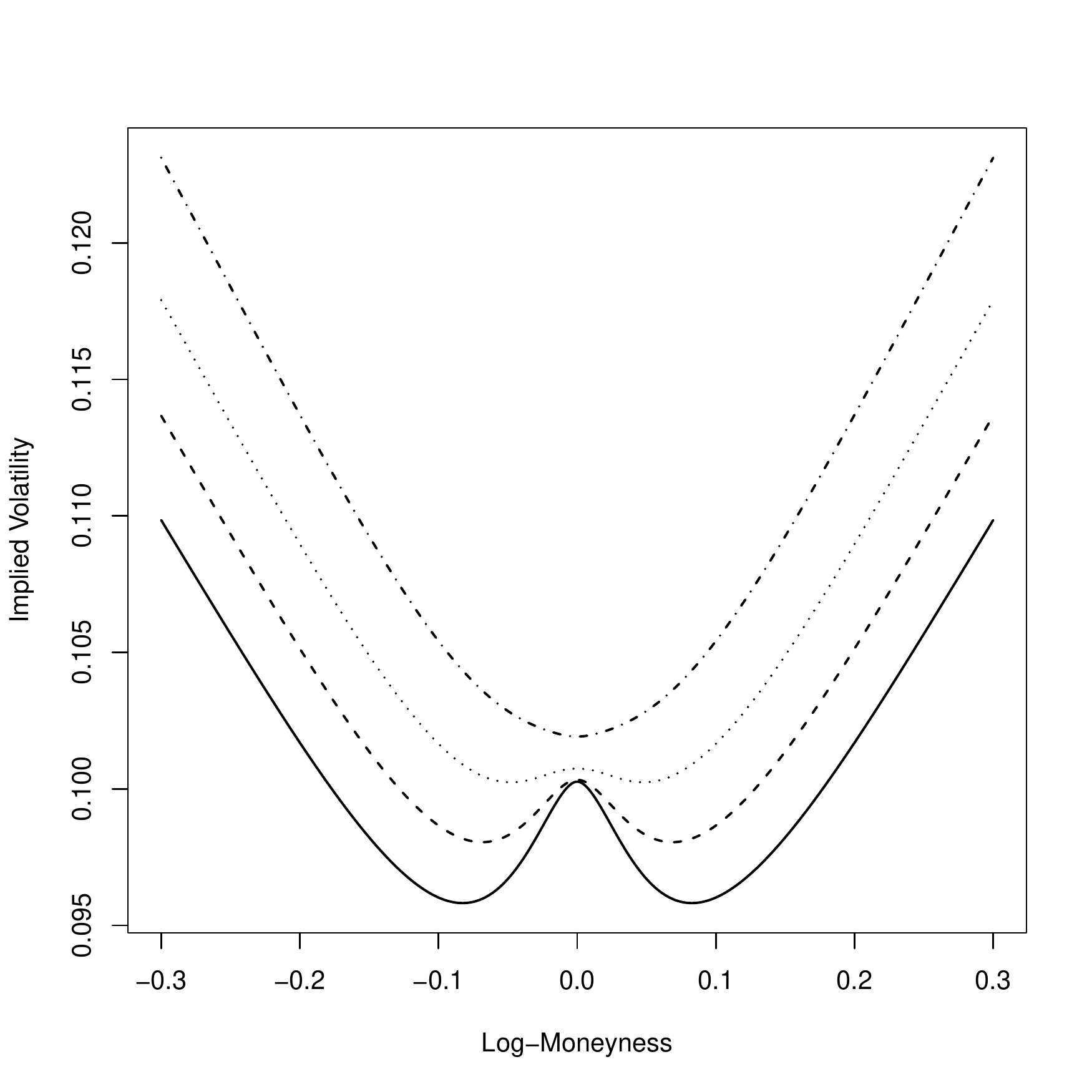}
\captionof{figure}{\textbf{W-shaped implied volatility curves in the variance-gamma mixture model.} This plot shows implied volatility curves for the variance-gamma mixture model with varying parameter $v = (0, 0.01, 0.015, 0.02)$, ordered from bottom to top. The volatility curves are in one-to-one correspondence with the densities shown in Fig.~\ref{fig:density}. The solid line corresponds to $v=0$, i.e., the double gamma model.}
\label{fig:IV}
\end{center}

\subsection{Option pricing and numerical examples}

The following formulas for put and call prices can be easily derived from \eqref{eq:call} and the explicit form of the densities $f_v$ in \eqref{eq:VG_density}, \eqref{eq:mix_dens}, and \eqref{eq:double_gamma}.
\begin{lem}The call price in the variance-gamma mixture model \eqref{eq:mixture} is given by
\begin{equation}\label{eq:call_VG}
C(K) = S_0 Q(\log(S_0/K)) - K \bar Q(\log(K/S_0)),
\end{equation}
and the put price by
\begin{equation}\label{eq:put_VG}
P(K) = K Q(\log(K/S_0)) - S_0 \bar Q(\log(S_0/K)),
\end{equation}
where 
\[Q(x) = \tfrac{a}{a+b} F_{VG}(x; \alpha,\beta_-,cT) + \tfrac{b}{a+b} F_{VG}(x; \alpha,\beta_+,cT),\]
with $F_{VG}$ the probability distribution function of the variance-gamma distribution in $(\alpha,\beta,c)$-parameterization, see \eqref{eq:VG_density}, and $\bar Q = 1 - Q$. In the limiting case $v = 0$, the formulas \eqref{eq:call_VG} and \eqref{eq:put_VG} holds with $Q$ given by
\[Q(x) = \tfrac{a}{a+b} \left(1 - \Gamma(-x; \lambda_-,cT)\right) + \tfrac{b}{a+b} \Gamma(x; \lambda_+,cT),\]
where $\Gamma$ is the probability distribution function of the Gamma distribution, extended by zero to $x < 0$.
\end{lem}

Using the package \textsf{VarianceGamma} \cite{scott2018variance} we have implemented the above formulas in the statistical computing environment \textsf{R} \cite{R2019} and plotted the implied volatility curves shown in Figure~\ref{fig:IV}. The plots nicely illustrate the symmetry, the obtainable W-shapes, and the dependence on $v$ in the variance gamma mixture model \eqref{eq:mixture}.


\small

\bibliographystyle{alpha}
\bibliography{references}


\appendix
\section{Lemma on uniform convergence of densities}
\begin{lem}\label{lem:analytic}The density $f_v(x)$, see \eqref{eq:mixture}, of the variance-gamma mixture model can be analytically extended to the right complex half plane $H_+ = \set{u \in \CC: \Re\,u > 0}$ and to the left complex half plane $H_- = \set{u \in \CC: \Re\,u < 0}$; the same is true for the double gamma density $f_0(x)$. On both $H_+$ and $H_-$, the extended densities $f_v$ converge uniformly on compacts to $f_0$ as $v \to 0$.
\end{lem}
\begin{proof} By \cite[9.6.1]{abramowitz1964handbook} the modified Bessel function $K_\nu$ is analytic in the complex plane cut along the negative real axis, from which the claim on analytic extension follows in conjunction with \eqref{eq:VG_density}. To show the compact convergence, we decompose $f_v$ as in \eqref{eq:mix_dens}, and first show the compact convergence of $f_{v,+}$ to the (analytically extended) Gamma density $f_{0,+}$ on $H_+$. We write $f_{v,+}(u) = A_v(u) \cdot B_{\alpha_v}(u)$, where
\[A_v(u) = \frac{(\lambda  + \tfrac{\mu}{2})^{cT}}{\Gamma(cT)} \left(\frac{u}{\alpha_v v}\right)^{cT - 1} e^{(\beta_+ - \alpha_v)u} \qquad \text{and} \qquad B_\alpha(u) = \sqrt{\frac{2 \alpha u}{\pi}} e^{\alpha u}K_{ct - 1/2}(\alpha u).\]
Using standard arguments, together with \eqref{eq:VG_par} and \eqref{eq:alpha_beta}, it is easy to see that
\[A_v(u) \to A_0(v) := \frac{\lambda_+^{cT}}{\Gamma(cT)} u^{cT-1} e^{-\lambda_+ u} \qquad \text{uniformly on compact subsets of $H_+$.} \]
It remains to show compact convergence of $B_{\alpha_v}$ to $1$ on $H_+$ as $\alpha_v \to \infty$.
To this end we use the integral representation \cite[9.6.24]{abramowitz1964handbook} of $K_\nu$ and the half-angle formula for hyperbolic functions, to write $B_\alpha(u)$ as
\[B_\alpha(u) = \sqrt{\frac{2 \alpha u}{\pi}} \int_0^\infty e^{-2\alpha u \sinh(t/2)^2} \cosh(\nu t) dt, \]
where $\nu = cT - 1/2$. For $\nu = 1/2$ we have $\sqrt{\tfrac{2 \alpha u}{\pi}} e^{\alpha u} K_{1/2}(\alpha u) = 1$, which allows us to write
\begin{equation}
B_\alpha(u)  - 1 = \sqrt{\frac{2 \alpha u}{\pi}} \int_0^\infty e^{-2\alpha u \sinh(t/2)^2} \left(\cosh(\nu t)  - \cosh(t/2)\right) dt.
\end{equation}
Let now $K$ be a compact subset of $H_+$. Clearly, there are constants $C_i > 0$, such that $|u|/\Re\,u  \le C_1$, $1/\Re\,u \le C_2$, and $\Re\,u \le C_3$ for all $u \in K$. We estimate
\begin{align}
\left|B_\alpha(u)  - 1\right| &\le C_1  \sqrt{\frac{2 \alpha \Re\,u}{\pi}} \int_0^\infty e^{-2\alpha \Re\,u \sinh(t/2)^2} \left|\cosh(\nu t)  - \cosh(t/2)\right| dt. \notag \\
& = C_1 \sqrt{\frac{2}{\pi}}\int_0^\infty e^{-y^2/2} \left| h\left(y/(2\sqrt{\alpha \Re\,u}\right) - 1\right| dy,
\label{eq:B_diff}
\end{align}
where
\[h(z) = \frac{\cosh(2 \nu \arsinh(z)}{\cosh(\arsinh(z))}.\]
Using $\arsinh(z) = \log(z + \sqrt{z^2 + 1})$, we find that $\log(z+1) \le \arsinh(z) \le \log(2z+1)$ and estimate
\begin{align*}
\sup_{u \in K} \left| h\left(y/(2\sqrt{\alpha \Re\,u}\right) - 1\right| &\le \left(\sqrt{\tfrac{C_2}{\alpha}}y + 1\right)^{2 \nu} - 1\\
\sup_{u \in K} \left| h\left(y/(2\sqrt{\alpha \Re\,u}\right) - 1\right| &\ge \left(\tfrac{1}{\sqrt{C_3 \alpha}}y + 1\right)^{-1} - 1
\end{align*}
Sending $\alpha \to \infty$, both bounds tend to zero; applying dominated convergence to the right side of \eqref{eq:B_diff} now shows that 
\[\lim_{\alpha \to \infty} \sup_{u \in K} \left|B_\alpha(u)  - 1\right| = 1,\]
i.e., $B_\alpha$ converges compactly to 1 on $H_+$, and hence also $f_{v,+}$ to $f_{0,+}$. Slight modifications of the above argument show that $f_{v,-}$ compactly converges to $f_{0,-}$ on $H_-$, and that both $f_{v,+}$ and $f_{v,-}$ compactly converge to zero on $H_-$ and $H_+$ respectively. Together with \eqref{eq:mix_dens}, the proof is complete.
\end{proof}

\end{document}